\documentclass[12pt,letterpaper,final]{article}
\usepackage[american]{babel}
\usepackage[numbers]{natbib}
\usepackage[top=1in, bottom=1in, left=1in, right=1in]{geometry}

\usepackage{tikz}
\usetikzlibrary{matrix}
\usetikzlibrary{positioning}
\usepackage{amssymb,amsmath,graphicx,mathtools, verbatim,
  caption,lscape} 
  \usepackage{booktabs,multirow} 
  \usepackage[nodayofweek]{datetime}
\mdyydate
\usepackage{relsize}\usepackage[printonlyused,withpage,smaller]{acronym}
\usepackage{varioref} \usepackage[hidelinks]{hyperref}
\usepackage[amsmath,amsthm,thmmarks]{ntheorem}
\usepackage{nameref}
\usepackage{thmtools}
\usepackage[capitalize]{cleveref} 
\usepackage{hyperref}

\theoremstyle{plain}
\newtheorem{theorem}{Theorem}[section]

\theoremstyle{plain}
\newtheorem{proposition}[theorem]{Proposition}

\theoremstyle{plain}
\newtheorem{lemma}[theorem]{Lemma}

\theoremstyle{plain}

\theoremstyle{plain}
\newtheorem*{corollary}{Corollary}

\theoremstyle{plain}
\crefname{open-problem}{Open Problem}{Open Problems}
\newtheorem{open-problem}[theorem]{Open Problem}

\theoremstyle{plain}
\crefname{result}{Result}{Results}

\theoremstyle{definition}
\newtheorem{definition}{Definition}[section]

\theoremstyle{definition}

\theoremstyle{definition}

\theoremstyle{definition}

\theoremstyle{definition}
\crefname{axiom}{Axiom}{Axioms}
\newtheorem{axiom}{Axiom}[section]

\theoremstyle{remark}
\theoremsymbol{\ensuremath{\lhd}}
\newtheorem*{remark}{Remark}

\theoremstyle{remark}
\theoremsymbol{\ensuremath{\lhd}}
\newtheorem*{note}{Note}

\theoremstyle{remark}
\theoremsymbol{\checkmark}

\theoremstyle{remark}
\theoremsymbol{\ensuremath{\blacksquare}}
\newtheorem*{indeed}{Indeed}

\crefname{equation}{Equation}{Equations}
\crefname{figure}{Figure}{Figures}

\usepackage{xargs}

\newcommand{\p}{\phantom{-}}

\usepackage[printwatermark=false]{xwatermark}
\newsavebox\mywatermarkbox\savebox\mywatermarkbox{  \tikz[color=red,opacity=0.15]\node[text width=12cm,align=center]{        arXiv v1: \today\\
        J. N. Clark
  };}
\newwatermark*[
  allpages,
  angle=60,
  scale=6,
  xpos=-0,
  ypos=0
]{\usebox\mywatermarkbox}

\makeatletter%
\begin{document}
\newcommand{\dissertationtitle}{Shapley-like values without symmetry}
\title{\dissertationtitle}

\author{Jacob North Clark\thanks{Department of Mathematics, University
    of Missouri-Columbia,~\texttt{jnclark@mail.missouri.edu}}\\Stephen
  Montgomery-Smith\thanks{Department of Mathematics, University of
    Missouri-Columbia,~\texttt{stephen@missouri.edu}}}
\date{}
\maketitle

\section*{\centering Abstract}

Following the work of Lloyd Shapley on the Shapley value, and
tangentially the work of Guillermo Owen, we offer an alternative
non-probabilistic formulation of part of the work of Robert J. Weber
in his 1978 paper ``Probabilistic values for games.''  Specifically,
we focus upon efficient but not symmetric allocations of value for
cooperative games. We retain standard efficiency and linearity, and offer an alternative condition,
``reasonableness,''  to
replace the other usual axioms. In the pursuit of the result, we discover
properties of the linear maps that describe the allocations. This
culminates in a special class of games for which any other map that is
``reasonable, efficient'' can be written as a convex combination of
members of this special class of allocations, via an application of
the Krein-Milman theorem.

\makeatletter{\renewcommand*\@makefnmark{}
  \footnotetext{\emph{Mathematics subject classification:} 91A12, 52A20}
  \footnotetext{\emph{Key words and phrases:} Shapley value,
    cooperative game, reasonable, efficient}
\makeatother}

\section{Introduction}\label{cha:introduction}

\subsection{History and Applications}

\subsubsection{The initial work}
In~\citeauthor{Shapley53}'s \citeyear{Shapley53} work, entitled
``A value for n-person games.'', Lloyd~\citeauthor{Shapley53} established an important idea in the theory of
collaborative games. In \citeauthor{Shapley53}'s own words,
``the possibility of evaluating games is therefore of critical
importance.'' A player in the game needs to know their
prospects, what they might receive compared to what they might
produce on their own. \citeauthor{Shapley53}'s work set forth an axiomatically
based way to do just that.

\subsubsection{Applications and iterations}
Since its first appearance, the Shapley value has been utilized in
numerous contexts.

One context in which the Shapley value appears is in social network
analysis. A person or organization might want to know who is the most important, or most
influential in a network.
In social contexts, one might want to impartially find the leader of a
community or rank the importance of members of a team.
In a strictly economic sense, this could be used to target demonstrations
or free samples of products, or used to target
advertising dollars to the ``taste makers'' of a network. More detail
regarding these ideas can be found in the work of
\citeauthor{GGMOPT03,NaNa11, PaGi11}~\cite{GGMOPT03,NaNa11, PaGi11},
and more generally the seminal work of
\citeauthor{Myerson77}~\cite{Myerson77} on ``Graphs and Cooperation in Games''.

Additionally, the Shapley value has been used in more general economic and political
applications. \citeauthor{Mertens02} has a compact
writeup, ``Some Other Economic Applications of the Value''~\cite{Mertens02} which discussed some
of these applications, such as taxation and redistribution, and
economies with fixed prices. Additionally, for voting games, the
Shapley-Shubik power index builds on the ideas of the Shapley value to
measure the power of each vote in voting games~\cite{ShaShu54},
something of interest to the field of political science among other fields. 

In many, if not all, cases in a usable context, the computations
necessary to calculate this information are
numerous, if not computationally prohibitive. As such, many
approximation schemes have appeared, as seen in the
papers by \citeauthor{Owen72,Fatima08} and \citeauthor{Castro09} in
\citeyear{Owen72,Fatima08} and \citeyear{Castro09}
respectively~\cite{Owen72,Fatima08,Castro09}. Algorithms
using linear and polynomial techniques have been considered, among others.

With all of this activity, one might question Shapley's initial
axioms. What is the fairness that his axioms describe? There have been
many explorations of variations of the Shapley 
value concept, such as probabilistic values and indices of power  as
summarized in ``Variations on the shapley value'' of \citeauthor{MonSam02}
\cite{MonSam02}. What happens when an axiom is weakened or removed? In
\citeauthor{Weber78}'s paper of \citeyear{Weber78},
``Probabilistic values for games,'' there was an initial investigation of some of
these ideas for the probabilistic value view of the Shapley value.
For other explorations of \citeauthor{Weber78}'s work,
\citeauthor{Derks05} \cite{Derks05} offers another proof.

This paper covers some of the same ideas as the above with alternative
assumptions placed on the allocations, aiming to make apparent the intermediary
implications of the work, as well as offering another method of proof
for the main result.

\subsubsection{Contributions of the paper}

In the paper~\cite{Weber78}, \citeauthor{Weber78} gave many results on
the theory of probabilistic values for games. One in particular is the
fact that one can characterize games that are efficient without
symmetry, i.e.\ random order values, as probabilistic
values~\cite[Theorems 12 and 13]{Weber78}. In this
paper, we offer an alternative idea and path to the results in
the world of these non-symmetric games, focusing upon intermediary
results, with the results below.

\begin{note}
  These results were inspired by the papers~\cite{Owen72,Shapley53}
  and without knowledge of~\cite{Weber78, Derks87}, until later on in idea
  development. The main difference between this paper, and the one of
  \citeauthor{Weber78} is we begin with more restricted, but
  reasonable assumptions, the Krein-Millman theorem is used
  in the proof of the main result, and properties of allocations
  themselves are looked into individually, rather than the whole
  process at once.
  Our process also differs from the presentation of \cite{Derks87}, in
  our focus upon the matrices and the Krein-Millman theorem as a
  vehicle to understand  the process.
\end{note}

\subsection{Review \& Definitions}\label{sec:literature-review}

To begin, we must first familiarize ourselves with the notion of an
\(n\)-person cooperative game in the style of
\citeauthor{Shapley53}~\cite{Shapley53}, or for a more modern
presentation see the exposition of  \citeauthor{Maschler13}~\cite{Maschler13}. In this
\lcnamecref{cha:introduction}, and the ones following, a cooperative game can 
be characterized as the following \lcnamecrefs{sec:literature-review} describe.

\subsubsection{Characteristic functions of games}

We begin with a set of players \(N = \{1,2,3,4,\ldots, n\}\), who may
or may not be cooperating with one another. For convenience, we denote
games with this number of players \(|N|\)-player games, or more
commonly \(n\)-player games. With our set of players, we now endeavor
to find a convenient way to mathematically express the possible gains
that various subsets of players would receive if they collaborated. To
accomplish this, we utilize characteristic functions. 

\begin{definition}
  Given an \(n\)-player cooperative game, with players coming from the
  set \(N\), we characterize the game in
  terms of possible collaborations, via its \emph{characteristic function}
  \(v\), where
  \(
    v:\mathbb{P}\left(N\right) \rightarrow
    \mathbb{R}^{\geq 0} 
  \)
  or, alternatively the domain is \(\{0,1\}^{|N|}\),
  i.e.\ in each situation, either a player is participating in a
  collaboration, or not, and the characteristic function assigns some
  value, or ``gains'' to this collaboration.
\end{definition}

Now, we wish to obtain new information about these
characteristic functions. First off, one may view them as a vector,
with each entry in the vector corresponding to a member 
\(T \in \mathbb{P}\left(N\right)\), applying some logical
ordering scheme to the vector, such as increasing cardinality from the
top to bottom of the vector. This vector view of a characteristic
function will be useful in the
considerations to come.

A characteristic function can exhibit several useful properties,
described below. 

\begin{definition}[Monotonicity]\label{def:monotone}
  A characteristic function \(v\) is called \emph{monotone} if given sets
  \(S\) and \(T\), with
  \(S\subseteq T\), then
  \(
    v(S) \leq v(T).
  \)
\end{definition}

\begin{definition}[Superadditivity]\label{def:superadditvity}
  A characteristic function \(v\) is called \emph{superadditive} if for
  all \(S, T \subset N\), if \(S\cap T = \varnothing\), then
  \(
    v(S\cup T) \geq v(S) +v(T).
  \)
\end{definition}

\subsubsection{Shapley's value and the Collaborative Game}

With this information about the game, we now shift focus to that of
allocating the spoils of the collaboration to each player. Typically,
this solution is viewed as a vector, \(\phi(N;v)\) and the gains
assigned 
to each player are denoted \(\phi_i(N;v)\) for player \(i\). One can
call 
this \(\phi\) an allocation.

\begin{definition}
  An \emph{allocation} is a function \(\phi\) with two inputs, the set of all
  players \(N\), and a characteristic function \(v\) with players from
  \(N\), with the output of the function in \(\mathbb{R}^{|N|}\).
\end{definition}

The familiar
Shapley value is one such allocation. To arrive at the Shapley value,
we need 
to familiarize ourselves with his axioms for a ``fair'' solution
\(\phi\) to the problem of dividing spoils. For the following, let us assume \(v(\varnothing) = 0\).

\begin{axiom}[Efficiency]\label{ax:eff}
  An allocation \(\phi\) is \emph{efficient} if for every
  coalitional game \((N;v)\), 
  \[\sum_{i \in N} \phi_i(N;v) = v(N).\] 
\end{axiom}

\begin{definition}
  Let \((N;v)\) be a coalitional game, and let \(i, j \in N\). Players
  \(i\) 
  and \(j\) are \emph{symmetric} if for every coalition 
  \(S \subseteq N \setminus \left\{i,j\right\} \), we have 
  \(v\left(S\cup \left\{i\right\}\right) = v\left(S \cup \left\{j\right\}\right).\)
\end{definition}

\begin{axiom}[Symmetry]\label{ax:sym}
  An allocation \(\phi\) is \emph{symmetric} if for every coalitional
  game \((N;v)\) and every pair of symmetric players \(i\) and \(j\)
  in the game: 
  \(\phi_i(N;v) = \phi_j(N;v)\)
\end{axiom}

\begin{definition}
  A player \(i\) is called a \emph{null player} in a game \((N;v)\) if for
  every coalition \(S \subseteq N\), including the empty coalition
  one has 
  \(v(S) = v(S\cup \left\{i\right\}) \).
\end{definition}

\begin{axiom}[Null player property]\label{ax:npp}
  An allocation \(\phi\) satisfies the \emph{null player property}
  if 
  for every coalitional game \((N;v)\) and every null player \(i\)
  in the game,
  \(\phi_i(N;v) = 0.\)
\end{axiom}

\begin{axiom}[Additivity]\label{ax:add}
  An allocation \(\phi\) satisfies \emph{additivity} if for every 
  pair of coalitional games \((N;v)\) and \((N;w)\),
  \(\phi(N;v+w) = \phi(N;v)+ \phi(N;w)\).
\end{axiom}

Putting together all of our axioms, we can finally obtain the Shapley
value. 

\begin{theorem}[Shapley value]
  There is a unique allocation \(\phi_i(N;v)\) satisfying efficiency, addativity, the
  null player property, and symmetry. 
  This is the Shapley value.
\end{theorem}

\begin{definition}
  The \emph{Shapley value} is given by the equation
  \[
    \phi_i(N;v) = \sum_{S\subseteq N\setminus\left\{i\right\}}
    \frac{|S|! (|N|-|S| - 1)!}{|N|!}\left(v(S\cup \left\{i\right\}) -
      v(S)\right).
  \]
\end{definition}

The Shapley value can also be determined via a path integral
calculation using a multi-linear extension of \(v\) as described
by~\citeauthor{Owen72}~\cite{Owen72}. This idea led, somewhat
tangentially, to the formulation and results of this paper.

\subsubsection{Analysis background}

In the proofs of our results, we invoke several analytical results. So,
to make the explanations clear, we present the results and
concepts from functional analysis we shall draw from.

\subsubsection{Extreme Points}

We familiarize ourselves first with the concept of extreme points.

\begin{definition}
  Let \(X\) be a vector space, and suppose \(K\) is a subset of   \(X\). A point \(x\in K\) is an \emph{extreme point} of \(K\) if it
  does not lie on a line segment in \(K\). To be more explicit, \(x\)
  cannot be written as a (generalized) linear combination of distinct
  values in \(K\). 
\end{definition}

We shall denote the set of extreme points of \(K\)
\(\text{ex}(K)\). Typically, we consider convex \(K\). 

Another way to view the definition of an extreme point \(x\),
following~\citeauthor{Bowers14}~\cite{Bowers14}, is if \(u\) and \(v\)
are elements of \(K\)
such that \(x=(1-t)u+tv\) for some \(t\in(0,1)\), then
\(x=u=v\). Namely, we cannot write an extreme point as the convex
combination of two distinct points in the set.

\subsubsection{Metrizable topological vector spaces}
\label{sec:metr-topol-vect}

Following the exposition by~\citeauthor{AlBo06}~\cite{AlBo06}, we explore some facts about
metrizable topological vector spaces, that will also be useful in
proving our results. (Although, we do not need the full power of any
of the statements.) 

\begin{definition}
  A \emph{neighborhood base at 0} is a collection of sets \(\mathcal{B}\)
  of neighborhoods of \(0\)  with the property that if \(U\)
  is any neighborhood of \(0\), there exists a \(B\in \mathcal{B}\)
  such that \(B\subset U\).
\end{definition}

\begin{theorem}
  A Hausdorff topological vector space is metrizable if and
  only if zero has a countable neighborhood base.
\end{theorem}

\begin{theorem}
  In a complete metrizable locally convex space, the closed convex
  hull of a compact set is compact.
\end{theorem}

\subsubsection{The Krein-Milman Theorem}

The Krein-Millman Theorem, of functional analysis, is yet another
result we shall utilize in our processes.

\begin{theorem}[Krein-Milman]\label{thm:krein-Millman}
  Suppose \(E\) is a locally convex Hausdorff topological vector space.
  If \(K\) is a nonempty compact, convex subset of \(E\), then
  \[
    K = \overline{\text{co}}\left(\text{ex}K\right)
  \]
  where \(\text{ex}\) is the set of extreme points, and
  \(\overline{\text{co}}\) is the closed convex hull.
  In particular, \(\text{ex}(K)\neq \varnothing\) 
\end{theorem}

The proof of the Krein-Millman Theorem is non-constructive,
however, the power of this result allows us to prove our results more
intuitively.

\section{Allocations of Value}
\label{cha:allocations-value}

\subsection{An introduction to allocations}

The Shapley value is a very specialized concept, and the given axioms
might be too specific in some situations. From this point forward, we
see if we can
generalize the idea of division of total value \(v(N)\) among the
players in \(N\), while reducing the number
of required axioms. In addition, we shall see what properties we can
determine based on these axioms.

The first thing we notice is the fact that these
allocations \(\phi\) can be viewed as linear maps, assuming we adopt
\cref{ax:add} (\nameref{ax:add}). As such, they can be viewed 
as \(|N| \times 2^{|N|}\) matrices. This view works quite well with
the vector view of the characteristic functions discussed
previously. Thus, we assume that \cref{ax:add} holds in all of our
further 
considerations. It is helpful to note that in our considerations,
\(N\) is fixed, so \(\phi\) can be viewed as a function of \(v\)
only. We also often make the identification between \(\phi\) and
its matrix counterpart. The majority of our results are proved using
this identification.

Inspired by the ideas presented by~\citeauthor{Owen72}~\cite{Owen72},
one can consider the path integrals along the edges of the region of
integration, rather than the main diagonal (corresponding to the
Shapley value). One may quickly see that the resulting allocations are
well behaved. They can be defined via set chains of the
players in \(N\), specifically set chains that contain all players
introduced one by one.  More formally, allocations are ``special'' if
they are the allocations described below.
\begin{definition}
  A \emph{special allocation} \(\phi\) is an allocation that
  assigns marginal contributions directly to players in the following
  way. Given a set chain
  \[
    \varnothing = M_0 \subset M_1 \subset M_2\subset \ldots \subset
    M_{|N|-1} \subset M_{|N|}=N
  \]
  with \(\left|M_{m+1}\setminus M_m\right| = 1\), player
  \(M_{m+1}\setminus M_m = \{i\}\) is assigned the ``gains''
  \[
    \phi_{i}(N,v) = v(M_{m+1}) - v(M_m).
  \]
\end{definition}

Thinking about this in matrix form we can see there is
``special'' structure here as well. Using the set chain
made up of \(M_m\), with \(m=0\) to \(m=|N|\) starting with
\(M_0=\varnothing\) with the restriction for all integer \(m\) between
\(0\) and \(\left|N\right|-1\) that
\[
  \left|M_{m+1}\setminus M_m\right| = 1
\]
we see this means that we are adding a single player to \(M_j\) at
each step in the chain. This corresponds with the matrix of the
allocation directly, namely for \(m\) from \(0\) to
\(\left|N\right|-1\), we interpret the set chain as follows: in the
row for player \(M_{m+1}\setminus M_m =\{i\}\), we place a \(-1\) in
the column associated with \(M_m\) and a \(1\) in the column
associated with \(M_{m+1}\). Looking at this, one might see why such
allocations are ``special''.

We quickly see that there are nice consequences of viewing these
special allocations as matrices.
The sum of elements in the first column is \(-1\) and the sum of
elements in the last column is \(1\). All other columns sum to
\(0\). The sum of the absolute value of the row elements is \(2\), and
additionally the sum of the absolute value of the interior column
elements (not in the first or last column) is \(2\)
as well. We will see all these consequences appear again, in more
general context.

These chains are in turn in one to one correspondence with the set of
all permutations of \(|N|\) letters, which we can use to find the
number of such allocations.

With this new terminology,
we endeavor to build a representative set of allocations from which we can
write any allocation. It turns out, we can do so for so-called
``reasonable, efficient'' allocations, and it turns out that this
set of representatives  is the set of all the special allocations. 

\subsubsection{Efficiency in allocations}

We offer a slight generalization of Shapley's efficiency useful to out
situation to begin

\begin{axiom}\label{def:newEfficiency}
  An allocation \(\phi\) with matrix \(A\) to be
  \emph{efficient} if
  \[
    \sum_{i=1}^n \phi_i(N;v) = v(N) -v(\varnothing)
  \]
  for all monotone \(v\), where \(\phi_j(N;v) = A_j \cdot v\), with
  \(A_j\) denoting the \(j^{th}\) row of \(A\).
\end{axiom}

\begin{note}
  For most practical applications, we can assume that \(v(\varnothing
  )=0\). With this assumption, we will divide all the spoils, even
  those present when no work is done by any player, among those
  participating in the game. Thus, we could think of efficiency as
  \[
    \sum_{i=1}^n \phi_i(N;v) = v(N) -v(\varnothing) = v(N).
  \]
\end{note}
This is of course closely related to \cref{ax:eff} (\nameref{ax:eff}).

It turns out that the row-wise sum properties we observed in the
special allocation's matrices are true of any efficient one.

\begin{lemma}\label{lem:rowwisesum}
  Column-wise, the sum of all the row elements in each column of the
  matrix of an
  efficient allocation is
  \[
      {\left(-1,0,\ldots,0,1\right)}.
  \]
\end{lemma}

\begin{proof}
  Suppose we have a \(n\) player game, with set of players \(N\).
  Let us also suppose we have an efficient allocation \(\phi\), with
  matrix \(A\). 
  By definition,
  \[
    A\cdot v = \phi(N;v)
  \]
  for all \(v\),
  with \(\phi_j(N;v)\) being the allocation of value to each player.
  Taking this information, we can now multiply both sides of the
  equality by a row univector of length \(n\), and obtain
  \[
    [1, \ldots, 1] A\cdot v = [1, \ldots, 1] \phi(N;v).
  \]
  Now, taking the right hand side, we notice
  \[
    [1, \ldots, 1] \phi(N;v)= \sum_{j=1}^n \phi_j(N;v).
  \]
  Via efficiency,
  \[
    \sum_{j=1}^n \phi_j(N;v) = v(N) - v(\varnothing).
  \]
  Of course,
  \[
    v(N) - v(\varnothing) = [-1, 0, \ldots, 0, 1] v.
  \]
  Putting all of this together,
  \[
    [1, \ldots, 1]A\cdot v = [-1, 0, \ldots, 0, 1] v.
  \]
  As this is true for all \(v\), we obtain
  \[
    [1, \ldots, 1]A = [-1, 0, \ldots, 0, 1].
  \]
  Therefore, the sum of the rows of \(A\) is what we require.
\end{proof}

\begin{note}
  The converse of this result is trivially true, namely if the sum of
  all the row elements in each column is
  \[
    \left(-1 , 0 , \ldots , 0 , 1\right)
  \]
  then the allocation is efficient. 
\end{note}

\subsubsection{Reasonableness in allocations}

We begin our study in earnest by by proposing a new axiom,
``reasonableness''.

\begin{axiom}\label{def:reasonable}\label{ax:reasonable}
  An allocation \(\phi\) with matrix \(A\)
  \emph{reasonable}\footnote{This 
    condition clearly implies several other conditions sometimes used
    in the explorations of allocations, namely \citeauthor{Weber78}'s
    \emph{dummy axiom}, and the \emph{null-player property}
    (\cref{ax:npp}). Recall, the dummy axiom is a
    generalization of the null player property.
    A player \(m\) is \emph{dummy} in the game if
    \[
      v(S\cup \{i\}) = v(S) + v(\{i\}) \text{ for all } S \subset
      N\setminus\{i\}.
    \]
    The dummy axiom is simply
    if player \(i\) is a \emph{dummy} in the game \(v\), then
    \(\phi_i(v)=v(\{i\})\).
  }
  if for all monotone \(v\),
  \begin{equation}\label{def:reasonablebounds}
    \min_{S: i\notin S}\left\{ v\left(S \cup\{i\}\right) -v(S)\right\}
    \leq A_i\cdot v = \phi_i(N;v) \leq \max_{S: i \notin
      S}\left\{v\left(S \cup \{i\}\right) - v(S)\right\},   
  \end{equation}
  where the maximum and minimum are taken over all \(S\), with \(i\notin S\).  
\end{axiom}

Why one might say this is ``reasonable'' is clear. A logical player
in a game would not expect to get less than the smallest contribution
they make to a group. In the same way, an impartial observer of a game
would not expect a player to receive more than the maximum
contribution a player made to any collaboration.

\begin{lemma}\label{lem:maxplayergains}
  Given a player \(m\), there exists a monotone \(v\) so
  \[
    \min_{S:m \notin S}\left\{ v\left(S\cup\{m\}\right) - v(S)\right\} =
    \max_{S:m\notin S}\left\{v\left(S \cup \{m\}\right) - v(S)\right\}.    
  \]
  i.e., the inequalities in \cref{ax:reasonable} are equalities.
\end{lemma}

\begin{proof}
  To begin, take \(m\) to be a given player in your game. We wish to
  build a binary characteristic function (or simple game) \(v_m\) so the minimum is equal
  to the maximum. We construct \(v_m\) as follows: If \(m \in S\) for each
  place, put 1 in that place, if not, place a 0. By construction, this
  is monotone. Also by construction, the difference
  \[v_m\left(S\cup\{m\}\right)-v_m(S) = 1\] for all \(S\) with
  \(S\cap\{m\} = \varnothing\). 
\end{proof}

\begin{note}
  This tells us
  \(\phi_m(N;v_m) = 1\), and \(\phi_l(N;v_m) = 0\) for \(l \neq m\).
\end{note}

The next two results follow quickly from the definitions.

\begin{lemma}\label{lem:charFuncSuperadd}
  The characteristic functions constructed in
  \cref{lem:maxplayergains} are %
  superadditive, as
  defined in \cref{def:superadditvity}.%
\end{lemma}

\begin{proposition}\label{prop:convexComboMaps}
  The convex combination of two reasonable allocations is again
  reasonable.
\end{proposition}

\begin{lemma}\label{lem:order1s}
  Given a matrix of a reasonable allocation \(A\) with each row
  only containing a single -1 and a single
  1 and the rest of the entries all being 0, if a -1 falls in the
  column for a set \(S\), then the associated 1 in that row must fall
  in a superset of \(S\), \(S \cup T\), where \(|T|-|S\cap T| > 0\).
\end{lemma}

\begin{proof}
  Suppose to the contrary, we have an allocation matrix \(A\), for
  which the 
  ordering is in reverse, namely the \(-1\) is in \(S\cup T\) and
  \(1\) in \(S\), in the row associated with player \(j\).
  Build \(v_S\) as the monotone function
  with \(v_S(S) = 0\) and \(v_S(S \cup T) = 1\) for all \(T\) such
  that  \(S \cap T \neq \varnothing\). This function is
  monotone by construction. When we utilize our map \(A\) to
  determine how to split the spoils, there is an immediate
  contradiction. For the player \(j\), \(\phi_j(v) = -1\). This
  contradicts the fact that reasonable allocations do not assign 
  players 
  negative spoils. At the very worst, for a monotone binary game \(v\),
  \[
    0\leq \phi_i(v,N)\leq 1 
  \]
  which is non-negative.
\end{proof}

\subsection{Extreme points of the reasonable efficient allocations}
By observation, with small sets of players one might infer that the
special allocations are the set of extreme points for the reasonable,
efficient allocations, see for example the appendices of \cite{ClarkDis}
for some finite results. We now
proceed to provide a proof of this assertion in general.

\subsubsection{Special allocations and sets}
\label{sec:saav-sets}

To begin, recall we defined each special allocation based on a
strictly increasing (by exactly one member at each step) chain of
sets. For example, the
following special allocation matrix for the game with \(3\) players,
\[
  \begin{pmatrix}
    -1&\p1&\p0&\p0&\p0&\p0&\p0&\p0\\
    \p0&-1&\p0&\p0&\p1&\p0&\p0&\p0\\
    \p0&\p0&\p0&\p0&-1&\p0&\p0&\p1
  \end{pmatrix}
\]
ordered usually, as follows
\[
  {\left[
    \varnothing , \{1\}, \{2\}, \{3\}, \{1,2\}, \{1,3\}, \{2,3\}, \{1,2,3\}
  \right]}^t,
\]
is associated with the chain of sets
\begin{equation*}
  \varnothing \subset \{1\} \subset \{1,2\} \subset \{1,2,3\}.
\end{equation*}
This idea of set chains and their connection to the special allocations will prove
integral to our following arguments. 
Additionally, it allows us to see the cardinality of the set of all
special allocations quite
quickly, as we know that the number of these chains corresponds
directly to the number of permutations of \(n\) letters, that is \(n!\).%

To prove our result, we will suppose we have an extreme reasonable,
efficient allocation not
listed in the set of all special allocations. We 
must show this is impossible. To obtain this result, we must first
note that 
any reasonable, efficient allocation has the following properties.

\subsubsection{Structural constraints on reasonable efficient allocations}

Each reasonable, efficient allocation has many properties, as
demonstrated previously, and shown explicitly in
\cite{ClarkDis}. Via those requirements, we
can find even 
more structure that will help us reach our result. To begin, we
will find the row-wise paring of elements. To make the proof more clear, the
following notation is introduced. 

\begin{definition}\label{def:matrixEntries} 
  Each entry in the matrix of an allocation can be referred
  to by a player and set, given \(A\),  we denote each entry by
  \[
    A_{i,S}
  \]
  where \(i\) denotes the player (row), as before, and \(S\) denotes
  the column associated.
\end{definition}

\begin{note}
  \Cref{def:matrixEntries} allows us to prove things
  independent from the ordering of the sets making up the columns of
  our matrices. Thus, if we can prove the statement that follows for a
  single row, we have it for all rows.
\end{note}

This allows us to write the payout to any specific player simply, as
\[
  \phi_i(v,N) = A_i\cdot v = \sum_{S\subset N} A_{i,S}v(S),
\]
recalling \(A_i\) is the \(i^{th}\) row of \(A\), and the sum is taken
over all \(S\subset N\).  

\subsubsection{Truncations of characteristic functions}

\begin{definition}
  We call a set \(S\) \emph{minimal} in the sense of the characteristic
  function if
  \(v(S)>0\) and there exists no set \(T \subsetneq S\) with
  \(v(T)>0\).
\end{definition}

\begin{definition}
  A \emph{truncation}\footnote{This is similar in spirit to~\citeauthor{Weber78}'s
    \emph{deletion}~\cite[Section 6]{Weber78}. The language truncation
    remains as we are thinking of these characteristic functions as
    vectors.} of a characteristic
  function \(v\) is the characteristic function \(w\) such that
  \(w(S) = 0\) for some minimal \(S\), and \(w(T) = v(T)\) for \(T\neq S\).
\end{definition}

\begin{remark}
  We call \(S\) the truncating set. 
\end{remark}

\begin{lemma}
  If \(v\) is monotone, then any truncation of \(v\) is monotone.
\end{lemma}

\begin{proof}
  Suppose \(v\) is a monotone characteristic function. Then, we know
  \(v(T) \geq v(S)\) for all \(S\subseteq T\) by definition. Let us
  let \(w\) be a truncation of \(v\), with truncating set
  \(S_t\). Recall, \(v(T) = w(T)\) for all \(T\neq S_t\), 
  and our inequality stands without much work for the majority
  of our places. However, we must concern ourselves of the cases when
  \(S_t\) appears, as \(v(S_t) > w(S_t) = 0\). This is no obstacle for
  \(T\) with \(S_t \subset T\), as
  \[
    w(T) = v(T) \geq v(S_t) > w(S_t)
  \]
  and
  \[
    v(S_t) > w(S_t)=0= v(S) = w(S)
  \]
  for \(S\subset S_t\). 
  Recall, of course, the truncating set is minimal, and there are no
  subsets with \(w(S)>0\). 
  Therefore, any truncation of \(v\) is again monotone.
\end{proof}

\begin{lemma}
  If \(v\) is superadditive, then any truncation of \(v\) is superadditive.
\end{lemma}

\begin{proof}
  Suppose \(v\) is a superadditive characteristic function. Then, we know
  \(v(S\cup T) \geq v(S) +v(T)\) for all \(S\) and \(T\) with
  \(S\cap T = \varnothing\), by definition. Let us let \(w\) be a
  truncation of \(v\), with truncating set \(S_t\). Recall,
  \(v(T) = w(T)\) for all \(T\neq S_t\), so the inequality
  stands without much work for the majority of our places. However, we
  must concern ourselves with the cases when \(S_t\) appears, as
  \(v(S_t) > w(S_t) = 0\). Note, however
  \[
    w(S_t\cup T) = v(S_t\cup T) \geq v(S_t) +v(T) > w(S_t) +w(T)
  \]
  if \(T\neq \varnothing\), and the statement is trivial if \(T\)
  is empty. Therefore, any truncation of \(v\) is again superadditive.
\end{proof}

\begin{definition}
  A \emph{pair truncation} of the characteristic function \(v\) is two successive
  truncations of \(v\) with truncating sets \(S\) and \(S\cup\{p\}\)
  respectively, for a player \(p\) with \(p\notin S\).
\end{definition}

\begin{lemma}\label{lem:truncFunc}
  A pair truncation \(w\) of a binary characteristic function \(v\) with marginal
  contribution of player \(p\) equal to 0 with truncating sets \(S\)
  and \(S\cup\{p\}\), where \(S\) is any minimal set with
  \(p\notin S\), retains the same 0 marginal contribution for p.
\end{lemma}

\begin{proof}
  To begin, let us take a characteristic function \(v\) with the marginal
  contribution of \(p\) equal to 0. Take the pairwise truncation of
  this \(v\) with the truncating sets \(S\) and \(S\cup\{p\}\), as
  described, where \(S\) is any minimal set with \(p\notin S\).
  Recall for \(v\), \(v(S)\) and \(v(S\cup \{p\})=1\).  Prior to pair
  truncation, we note that, for sets \(T \cap \{p\} =\varnothing\),
  \(v(T \cup \{p\}) - v(T) = 0\), as the marginal contribution of
  \(p\) is 0.  Naturally, following the pair truncation,
  \(w(T \cup \{p\}) - w(T)\) is still equal to \(0\) for all sets
  \(T\).  For all the unchanged places, this is clear, and for the two
  changed places, rather than seeing \(1-1=0\) for the case of
  \(T = S\), our truncating set, we now observe \(0-0=0\).
\end{proof}

\begin{note}
  This of course works on the characteristic functions one can
  construct following \cref{lem:maxplayergains} for player \(m\), and
  they all have the same (zero) marginal contribution for any player
  \(p \in N\setminus \{m\}\).
\end{note}

\subsubsection{Extensions of characteristic functions}

\begin{definition}
  Given a characteristic function \(v_M\)  on the set \(M\subsetneq N\), we can extend it to a
  characteristic function \(v_N\) on \(N\) by setting
  \[
    v_N(S) = v_M(S\cap M).
  \]
\end{definition}

\begin{lemma}
  If a characteristic function \(v_{M}\)  on \(M\subsetneq N\) is monotone, then its
  extension \(v_N\) to \(N\) is also monotone.
\end{lemma}

\begin{proof}
  As \(v_{M}\) is monotone, we have
  \[
    v_{M}(T) \geq v_{M}(S) 
  \]
  for all \(S\) and \(T\) with \(S\subset T\).
  For the extension, we note, as \(S\subset T\), \(S\cap M\subset
  T\cap M\)
  \begin{align*}
    v_{N}(T) &=v_{M}(T\cap M)\\
             &\geq v_M(S\cap M)\\
             &= v_N(S),
  \end{align*}
  and therefore, the extension is also monotone.
\end{proof}

\begin{lemma}
  If a characteristic function \(v_{M}\)  on \(M\subsetneq N\) is superadditive, then its
  extension \(v_N\) to \(N\) is also superadditive.
\end{lemma}

\begin{proof}
  As \(v_{M}\) is superadditive, we have
  \[
    v_{M}(S\cup T) \geq v_{M}(S) + v_{M}(T) 
  \]
  for all \(S\) and \(T\) with \(S\cap T = \varnothing\).
  For the extension, we note
  \begin{align*}
    v_{N}(S\cup T) &=v_{M}((S \cup T)\cap M)\\
                     &= v_{M}((S\cap M) \cup
                       (T\cap M))\\
                   &\geq v_M(S\cap M) + v_M(T\cap M)\\
                     &= v_N(S) + v_N(T),
  \end{align*}
  and thus, the extension is also superadditive.
\end{proof}

\begin{lemma}
  For the extensions of characteristic functions
  \(M=N\setminus\{i\}\), the inequalities in the definition of
  reasonableness taken for player \(i\)
  are equalities and
  \[
    \phi_i(N;v) = 0.
  \]
\end{lemma}

\begin{proof}
  This is clear, as \[v(S\cup \{i\}) = v(S)\]
  by the definition of an extension, and by the definition of reasonableness.
\end{proof}

\subsubsection{Pairing behavior in the rows of a reasonable allocation}

In the following theorem, we establish a strong condition on the
reasonable allocations, specifically their matrix counterparts. To our main
result, this theorem serves as an integral component.

\begin{theorem}[Pairing of row elements]\label{thm:rowWisePairing}
  Given a player (row) \(i\)  of a reasonable allocation's matrix
  \(A\), the elements pair off in the following manner:
  \begin{equation*}
    A_{i,S} = -A_{i,S\cup\{i\}}
  \end{equation*}
  for sets \(S\) with \(S \cap \{i\} = \varnothing\).
\end{theorem}

\begin{proof}
   To obtain this result, we consider a player \(i\). (Any other
   player's information can be obtained identically.)
  Construct the following superadditive \(v\), given \(S\) with
  \(S\cap \{i\} = \varnothing\)  
  \[
    v^S_a(T) =
    \begin{cases}
      1 & \text{if } S \subset T\\
      0 & \text{else}
    \end{cases}
  \]
  and
  \[
    v^S_b(T) =
    \begin{cases}
      1 & \text{if } S \subsetneq T\setminus\{i\}\\
      0 & \text{else}.
    \end{cases}
  \]
  Observe, \(v^S_a\) is superadditive trivially, as
  \[
    v^S_a(Q\cup R) \geq v^S_a(Q)+v^S_a(R)
  \]
  for \(Q\) and \(R\) with \(Q\cap R=\varnothing\). This can quickly
  be seen as \(S\subset Q\) or \(S\subset R\), but not both. So,
  at the very worst, \(1\geq 1+0\) or \(1\geq 0+1\).  

  Notice \(v^S_b\) is a pair truncation of \(v^S_a\), by construction, as
  it zeros out the \(S\) and \(S\cup\{i\}\) places precisely.
  Thus, \(v^S_b\)  is also superadditive. \(v^S_a\) also has the property that
  the marginal contribution of player \(i\) is 0, observe, by
  construction, for \(Q\) with \(Q\cap\{i\} = \varnothing\), 
  \[
    v^S_a(Q\cup\{i\}) -v^S_a(Q) =
    \begin{cases}
      1 - 1 = 0 & \text{if } S\subset Q\\
      0-0 = 0 & \text{if } S \not\subset Q
    \end{cases}.
  \]
  Recall, \(i\notin S\) by our initial choice of \(S\). So, to obtain
  the pairing for \(A_{i,S}\) and \(A_{i,S\cup\{i\}}\), we observe
  \begin{align*}
    \phi_i(N;v^S_a) = A_i \cdot v^S_a &= 0\\
    \phi_i(N;v^S_b) = A_i \cdot v^S_b &= 0.
  \end{align*}
  However, using this to our advantage, notice
  \begin{align*}
    0&= \phi_i(N;v^S_a) - \phi_i(N;v^S_b) \\
     &= A_i \cdot (v^S_a-v^S_b)\\
     &= A_{i,S} + A_{i,S\cup\{i\}}.
  \end{align*}
  Rearranging,
  \begin{equation}
    A_{i,S} = -A_{i,S\cup\{i\}},
  \end{equation}
  as we wished. Repeating this process for all
  \(S\subset N\setminus\{i\}\) gives us all the pairings we desire.
\end{proof}

\subsubsection{Bounds on the matrix elements}

\begin{lemma}\label{lem:orderSigns}
  Given a row \(i\) of a reasonable allocation's matrix \(A\), 
  for nonempty \(S\), with \(S\cap \{i\} = \varnothing\)
  \begin{equation*}
    0\leq A_{i,S\cup\{i\}} \leq 1
  \end{equation*}
  and 
  \begin{equation*}
    0\geq A_{i,S} \geq -1.
  \end{equation*}
  Namely, the non-negative entries fall in the columns associated with
  sets \(S\cup\{i\}\) and the associated non-positive entries fall in
  the columns associated with \(S\).
\end{lemma}

\begin{proof}
  Using techniques seen in the proof of \cref{thm:rowWisePairing}, we
  can make short work of this statement.
  Consider
  \[
    v^S_b(T) =
    \begin{cases}
      1 & \text{if } S \subsetneq T\setminus\{i\}\\
      0 & \text{else}
    \end{cases}
  \]
  as seen previously, and
  \[
    v^S_c(T) =
    \begin{cases}
      1 & \text{if } S \subsetneq T\\
      0 & \text{else}.
    \end{cases}
  \]
  Note, \(v^S_c\) is a truncation (not a pairwise truncation) of the \(v^S_a\) from the
  superadditive proof of \cref{thm:rowWisePairing}. Hence, it is
  superadditive. The only difference between \(v^S_b\) and \(v^S_c\) is
  \(v^S_b(S\cup\{i\}) = 0\), while \(v^S_c(S\cup\{i\}) = 1\). Recall, with
  \(v^S_b\), the marginal contribution of player \(i\) is always
  0. Note, also, for \(v^S_c\), the marginal contribution of player
  \(i\) falls between 0 and 1, as the only difference from \(v^S_b\)
  occurs at the place associated with \(S\), which results in
  \[
    v^S_c(S\cup\{i\}) -v^S_c(S) = 1-0 = 1
  \]
  rather than 0.
  Now, to use this to our advantage, we note
  \begin{align*}
    \phi_i(N;v^S_b) = A_i \cdot v^S_b &= 0.
  \end{align*}
  and
  \begin{equation*}
    0\leq \phi_i(N;v^S_c) = A_i \cdot v^S_c \leq 1
  \end{equation*}
  via reasonableness. However, there is little difference
  between \(A_i \cdot v^S_b\) and \(A_i \cdot v^S_c\), notice
  \begin{equation}\label{eq:collapsedRow}
    A_i \cdot v^S_c - A_i \cdot v^S_b = A_{i,S\cup\{i\}}.
  \end{equation}
  So, we may conclude
  \begin{align*}
    0\leq\phi_i(N;v^S_c) &- \phi_i(N;v^S_b) \leq 1\\
    0\leq A_i \cdot v^S_c &- A_i \cdot v^S_b \leq 1\\
  \end{align*}
  and utilizing \cref{eq:collapsedRow},
  \begin{equation}
    0\leq A_{i,S\cup\{i\}} \leq 1
  \end{equation}
  as we wished to show. By the pairings, we obtain
  \begin{align*}
    0\leq -A_{i,S} \leq 1,
  \end{align*}
  i.e.
  \begin{equation}
    0\geq A_{i,S} \geq -1
  \end{equation}
  again, as we wished to show.
\end{proof}

\begin{note}
  Via this lemma, we have the signs of nearly all of the matrix
  entries. The only missing are, for row \(i\) \(A_{i,\varnothing}\)
  and \(A_{i,\{i\}}\). This is covered by \cref{lem:lastcolumnBd}
  below.
\end{note}

\begin{lemma}\label{lem:lastcolumnBd}\label{lem:column1upperBd}
  Given a reasonable, efficient allocation matrix \(A\), we have
  \(1\geq A_{i,\{i\}} \geq 0\) for all players \(i\), and hence, by the pairing
  \(-1\leq A_{i,\varnothing} \leq 0\).
\end{lemma}

\begin{proof}
  We can utilize the information we have gained thus far to infer this
  information. First, recall that \(-1\leq A_{j,\{i\}} \leq 0\) for all
  \(j\neq i\), via the pairings, as by the previous
  \cref{lem:orderSigns}, \(1\geq A_{j,\{i,j\}} \geq 0\). More
  specifically,
  \[
    0\geq \sum_{j\neq i} A_{j,\{i\}} \geq -1.
  \]
  This is indeed the case, as if not, we shall reach a contradiction.
  Let us assume that
  \[
    \sum_{j\neq i} A_{j,\{i\}} \leq -1.
  \]
  Via the pairings, we can follow each of the nonzero elements in the
  sum up the chain to \(A_{j,\{i,j\}}\) and there must be a paired
  positive value there, for each \(j\). Their sum, even though they
  might not remain in the same column remains more than 1. In columns
  \(\{i,j\}\) for each of these \(j\), there must be a negative
  contribution to make the column-wise sum 0. (It is possible there
  might be a splitting between two or more rows. Keep in mind no entry
  can exceed 1 in absolute value by \cref{lem:orderSigns} at this
  point, if we have any entry greater than \(1\) in absolute value, we
  have a contradiction. The sum, of course must remain 0.) Each of
  those entries pairs off with a superset in the same row. We can
  continue this process until we reach the set of all players
  \(N\). In this column, we have the tail of every (possibly split up)
  chain we traveled along, and due to the pairing of the elements, the
  overall sum is greater than or equal to \(1\) (as each chain must
  carry at the very least all of its value along, even if it splits or
  combines along the way due to the pairings).  This contradicts
  reasonableness, and thus we have that
  \(\sum_{j\neq i} A_{j,\{i\}} \geq -1\) as we wished.
  Noting this, we see
  that it is imperative that \(A_{i,\{i\}}\) must be greater than or
  equal to \(0\) and less than or equal to \(1\), as, via efficiency,
  the sum of all elements in the column \(\{i\}\)
  is \(0\) by \cref{lem:rowwisesum}. More explicitly,
  \[
    \sum_{j\neq i} A_{i,\{j\}} + A_{i,\{i\}} = 0
  \]
  so
  \[
    A_{i,\{i\}} = -\left(\sum_{j\neq i} A_{i,\{j\}} \right)
  \]
  and
  \[
    0 \leq A_{i,\{i\}} \leq 1.
  \]
  Clearly, as a result,
  \(-1 \leq A_{i,\varnothing} \leq 0\) via the pairings.
\end{proof}

So, we now have the signs for all of the row elements, and we can
summarize our results in the following convenient way.

\begin{theorem}\label{thm:signRowElt}
  Given a reasonable, efficent allocation with matrix \(A\), for player \(i\),
  \begin{alignat*}{10}
    1 &\geq~A_{i,S}~\geq 0 &\text{ if } &i\in S \\
    -1 &\leq~A_{i,S}~\leq 0 &\text{ if } &i\notin S
  \end{alignat*}
\end{theorem}

\begin{proof}
  Refer to \cref{lem:orderSigns,lem:lastcolumnBd}.
\end{proof}

We utilize some of the ideas found in the proof of \cref{lem:column1upperBd} to get
more general results, along with
\cref{thm:signRowElt,thm:rowWisePairing}. So, as a consequence of
these results, we can find another important fact for our reasonable, efficient allocations.

\begin{lemma}\label{cor:colsum1}
  Each column of a reasonable, efficient allocation matrix can contain no
  more than a sum of \(-1\) of negative elements and a sum of \(1\) in
  positive elements.
\end{lemma}

\begin{proof}
  Suppose the contrary, that there is a sum of more than -1 of the
  negative elements in one column. Our aim is to show this is not
  possible. Let us assume that this overflow of negatives occurs in
  the column \(\{i,j\}\), in rows \(l, \ldots\), and \(m\), say.
  Now, via the
  pairing, we can follow each of these chains up, there must be a
  paired positive value in \(A_{l,\{i,j,l\}}, \ldots\), and \(A_{m,\{i,j,m\}}\)
  respectively, and their sum remains more than 1. In columns
  \(\{i,j,l\}, \ldots\), and \(\{i,j,m\}\), there must be a negative
  contribution to make the column-wise sum 0. (It is possible there
  might be a splitting between two or more rows. Keep in mind no entry
  can exceed 1 in absolute value, if we have any entry greater than
  \(1\) in absolute value, we have a contradiction. The sum, of course
  must remain 0.)
  Each of those entries pairs off with a superset in
  the same row. We can continue this process until we reach the set of
  all players \(N\). In this column, we have the tail of every
  (possibly split up) chain we traveled along, and due to the pairing
  of the elements, the overall sum is greater than or equal to \(1\)
  (as each chain must carry all of its value along). This contradicts
  reasonableness.

  The case of two or more positive values with sum greater than \(1\)
  is handled nearly identically, one just starts the chain at this
  point, noting the sum of the column must be 0, and thus, there must
  be elements summing to greater than \(-1\) in the column to
  compensate.

  Therefore, we have what we set out to show, each column can contain
  no more than a sum of \(-1\) of negative elements and a sum of \(1\)
  in positive elements.
\end{proof}

\begin{lemma}\label{thm:partialRowSum1}
  For a reasonable allocation \(\phi\) with matrix  \(A\), the
  partial row-wise sum satisfies the equality
  \[
    \sum_{S: i\notin S} A_{i,S \cup \{i\}} = 1.
  \]
\end{lemma}

\begin{proof}
  Recall, from \cref{lem:maxplayergains}, we know there exists a
  \(v_i\) for which 
  \[
    \phi_i(N,v_i) = 1
  \]
  by the squeezing of the reasonableness condition.
  This \(v_i\) is precisely the one with a \(1\) in all places with
  \(i\in T\). 
  Now, noting that
  \[
    \phi_i(N,v_i) = A_i\cdot v_i = \sum_{S: i\notin S} A_{i,S \cup \{i\}},
  \]
  we obtain the result we desire.
\end{proof}

\begin{corollary}[to \cref{thm:partialRowSum1}]
  For a reasonable allocation \(\phi\) with matrix \(A\), the
  partial row-wise sum satisfies the equality
  \[
    \sum_{S: i\notin S} A_{i,S} = -1.
  \]
\end{corollary}

\begin{remark}
  For a reasonable, efficient allocation with matrix \(A\), the sum over all
  players \(i\) of the elements \(A_{i,S} \geq 0\) with cardinality of
  \(S\) a fixed integer between 0 and \(|N|\) is
  \(\sum_i A_{i,S} = 1\).\footnote{This is similair to the
    implicit result alluded to in the proof of~\cite[Theorem
    13]{Weber78}.}
\end{remark}

\subsection{Extreme points, and the reasonable efficient allocations}

In the journey to prove our results, we find the following lemma
integral to our arguments as well.

\begin{lemma}\label{lem:altReasonable}
  For an allocation \(\phi\) with matrix \(A\), if for all
  players \(i\),
  \begin{equation}\label{eq:signElts}
    A_{i,S} \text{ is }
    \begin{cases}
      \geq 0 \text{ if } i\in S \\
      \leq 0 \text{ if } i\notin S
    \end{cases}.
  \end{equation}
  and for \(T\) with \(T\cap \{i\} = \varnothing\) 
  \begin{equation}\label{eq:pairningElts}
    A_{i,T} = -A_{i,T\cup\{i\}}
  \end{equation}
  and
  \begin{equation}\label{eq:partialRowSumElts}
    \sum_{T} A_{i,T\cup \{i\}}=1,
  \end{equation}
  then the allocation \(\phi\) is reasonable.
\end{lemma}

\begin{proof}
  To begin, by \cref{eq:partialRowSumElts,eq:pairningElts}, we have
  that \(\sum_{T} A_{i,T}=-1\), and all elements of the matrix are
  determined. To check for reasonableness,
  note 
  \begin{align*}
    \phi_i(N;v) & = A_i \cdot v\\
                & = \sum_S A_{i,S} \cdot v(S)\\
                & =  \sum_{T} A_{i,T\cup \{i\}}\left(v(T\cup \{i\}) - v(T)\right)
  \end{align*}
  with \(T\cap \{i\} = \varnothing\) by our map and
  \cref{eq:pairningElts}. By \cref{eq:partialRowSumElts,eq:signElts}, we note that
  \(\phi_i(N;v)\) is a generalized linear combination of the
  marginal contributions of each player, as all the elements in the
  sum are greater than or equal to 0 and sum to 1. 
  Trivially, a generalized convex
  combination lies between the
  \[
    \min_{T: i\notin T}\left\{ v\left(T \cup\{i\}\right) -v(T)\right\}
    \text{ and } \max_{T: i\notin T}\left\{ v\left(T \cup\{i\}\right)
      -v(T)\right\}.
  \]
  Thus, the allocation is reasonable, as required.
\end{proof}

\begin{lemma}
  The converse to \cref{lem:altReasonable} is also true.
\end{lemma}

\begin{proof}
  For the reverse, suppose matrix \(A\) and its associated
  \(\phi\)  is reasonable. In part, this 
  now equates to showing that the only possible choices for elements
  satisfy \cref{eq:signElts,eq:pairningElts,eq:partialRowSumElts}. The
  pairings, we note, can be found utilizing the results of
  \cref{thm:rowWisePairing}. Identically, the partial row sum in
  \cref{eq:pairningElts} is obtained via
  \cref{thm:partialRowSum1}. Both of these theorems use only the
  satisfaction of the reasonability condition in the body of their
  results. Taking note that we have the pairings, we can obtain almost
  all the signs we wish via \cref{lem:orderSigns}. However, in our
  proofs above, we cannot obtain the signs of \(A_{i,\varnothing}\)
  and \(A_{i,\{i\}}\) without efficiency. However, if we allow
  ourselves to utilize non-superadditive, yet monotone vectors, we
  can get the remaining signs we wish.
  To prove \(A_{i,\varnothing} \leq 0\) for all players \(i\), one can 
  observe, via multiplying the map \(A\) by the characteristic function
  \[v={\left[1,1,1,\ldots,1,1,1\right]}^t,\] focusing on the
  \(i^{th}\) row,  
  \begin{multline*}
    A_{i,\varnothing}+A_{i,\{1\}}+A_{i,\{2\}}+ \cdots 
    + A_{i,\{1,2\}}+\cdots +A_{i,N\setminus\{n\}} +\cdots
    +A_{i,N\setminus\{1\}}+A_{i,N} = 0
  \end{multline*}
  by reasonableness.  Similarly, we obtain
  \begin{multline*}
    1 \geq A_{i,\{1\}}+A_{i,\{2\}}+ \cdots 
    + A_{i,\{1,2\}}+\cdots +A_{i,N\setminus\{n\}} +\cdots 
    +A_{i,N\setminus\{1\}}+A_{i,N} \geq 0
  \end{multline*}
  by using the definition of reasonableness, \cref{ax:reasonable}, along with the vector
  \[v={\left[ 0,1,1,\ldots,1,1,1 \right]}^t.\]
  This is evident due to the fact that each player must receive no
  less than 0 and no more than 1, based on the marginal contribution
  bounds, given the fact our characteristic
  function is monotone. Combining these two facts, one quickly sees
  that
  \begin{equation}
    -1 \leq A_{i,\varnothing} \leq 0,
  \end{equation}
  as required. \(1 \geq A_{i,\{i\}}\geq 0\) is immediately picked up
  via the pairings. Therefore, we have the last piece of needed info,
  \cref{eq:signElts}.
\end{proof}

\begin{lemma}[Extreme points of the reasonable, efficient allocations]\label{lem:raavExtremePts}
  The extreme points of the set of all reasonable, efficient
  allocations are contained within   the set of special allocations.
\end{lemma}

\begin{proof}
  We prove this using the matricies of the allocations.
  Suppose first we have a reasonable, efficient allocation matrix \(A\) that is an extreme point,
  but is not a member of the set of all special allocations. Our aim is to reach
  a contradiction. To begin, via efficiency, and
  \cref{lem:rowwisesum}, we know there is at least one entry in the
  first column that is negative, in the \(i^{th}\) row, say. Starting
  at this entry, we build a set chain as introduced in
  \cref{sec:saav-sets} by utilizing the row-wise pairings 
  of \cref{thm:rowWisePairing}. Given our choice of negative element
  in the first column, \(A_{i,\varnothing}\), by the pairings, we know
  there is a positive, and equal in absolute value, entry in \(A_{i,\{i\}}\). Calling
  upon efficiency yet again, for the internal columns, we know that
  the sum of all the entries is equal to \(0\). Thus, there exists
  at least one negative element in row \(j\), say,
  \(A_{j,\{i\}}\). This entry in turn has a
  paired entry in a superset \(\{i,j\}\), \(A_{j,\{i,j\}}\). Continuing this
  process, we can continue to build a set chain to represent this path
  through the allocation. The set chain would appear as something of
  the form 
  \[
    \varnothing \subset \{i\} \subset \{i,j\} \subset \{i,j,k\}
    \subset \dots \subset N.
  \]
  If, at all stages the
  at least one player was exactly one player, we have a contradiction. Recalling
  \cref{cor:colsum1}, each column can 
  contain no more than a sum of \(-1\) in negative elements and a sum
  of \(1\) in positive elements. If there was only a single choice in
  each case, as the sum of the first column must be \(-1\), that
  forces \(A_{i, \varnothing}=-1\). In turn, \(A_{i,\{i\}}
  =1\). Continuing, it must be the case that \(A_{j,\{i\}}=-1\).
  Following this along the chain, we know every element we touched was
  either a \(-1\) or \(1\) by the pairings and efficiency
  (\cref{lem:rowwisesum}). More precisely, it was a special allocation already,
  a contradiction.  Therefore, we know that at in at least one
  instance when we were building our chain, we had two choices of
  elements (either both positive, or both negative, distinct from the
  element we started with), in row \(l\) and \(m\), say.
  If, in our initial chain, we chose row \(l\), we can make a secondary
  set chain by choosing row \(m\) at the juncture and following this alternative
  path. More explicitly, for some \(T\) with  \(T \cap\{l,m\} =
  \varnothing\) and \(l\neq m\), our two set chains would contain the links 
  \[
    T \subset T \cup \{l\}
  \]
  for the first set chain, and 
  \[
    T \subset T \cup \{m\}
  \]
  for the second set chain, respectively. Thus, we have two distinct
  set chains. From these two distinct set chains, we can
  find the associated special allocation matrices \(S^a\) and
  \(S^b\), say.
  These special allocations will be used to
  demonstrate that our assumed extreme \(A\) is
  not. Following the same ideas we did in the case of the
  characteristic functions, we
  wish to find a way to modify our \(A\) in small ways on either side,
  both modified matrices still reasonable, with a convex combination
  of the matrices equal to \(A\) itself. To do this, we choose an
  \(\epsilon\) in the following way:
  \begin{equation*}
    \epsilon < \min
    \begin{cases}
      A_{M_{m+1}\setminus M_{m}, M_{m+1}}\\
      1-A_{M_{m+1}\setminus M_{m}, M_{m+1}} & A_{M_{m+1}\setminus
        M_{m}, M_{m+1}} \neq 1
    \end{cases}
  \end{equation*}
  where \(M_m\) and \(M_{m+1}\) are consecutive elements in the set
  chains we have defined above. First, notice \(A_{M_{m+1}\setminus
    M_{m}, M_{m+1}} > 0 \) by construction.
  Notice also, as we have assumed \(A\) is not a special allocation, and have
  excluded the entries \(A_{M_{m+1}\setminus M_{m}, M_{m+1}} \neq 1\)
  from consideration in \(\epsilon\), we have \(0<\epsilon<1\). 
  By this choice of \(\epsilon\), we claim that one can
  both add and subtract \(\epsilon S^a \pm \epsilon S^b\) without
  compromising the reasonableness of the map.
  To see this, consider
  \[
    A \pm (\epsilon S^a - \epsilon S^b)
  \]
  alongside \cref{lem:altReasonable}. We note by our choice of
  \(\epsilon\), the signs of each element of
  \(A \pm (\epsilon S^a - \epsilon S^b)\) are unchanged.  Trivially,
  we satisfy the pairings for \(T\) with \(T\cap\{i\}=\varnothing\) as
  \(A\), \(S^a\) and \(S^b\) do as well, and addition and subtraction
  of matrices with the pairings produces other matrices satisfying the
  pairings. Finally, we note that each row in
  \(\pm(\epsilon S^a - \epsilon S^b)\) makes a contribution of 0 to
  the row sum of \(A' = A\pm(\epsilon S^a - \epsilon S^b)\), and the
  sum \(\sum_{T} A'_{i,T\cup \{i\}}=1\), as, in net, all that is done
  is an addition and subtraction of \(\epsilon\) to the sum. Thus, all
  of the requirements of \cref{lem:altReasonable} are fulfilled, and
  we can conclude that \(A'\) is reasonable. Clearly, as \(A\),
  \(S^a\) and \(S^b\) are efficient, the sum
  \(A' = A\pm(\epsilon S^a - \epsilon S^b)\) is efficient as well
  utilizing \cref{lem:rowwisesum} and its converse.
  Thus, we have two additional derived reasonable, efficient 
  maps, \(A_a\) and \(A_b\), with
  \begin{align*}
    A_a & = A - (\epsilon S^a - \epsilon S^b)\\
    A_b & = A + (\epsilon S^a - \epsilon S^b).
  \end{align*}
  Notice
  \[
    A = \frac{1}{2} A_a + \frac{1}{2} A_b.
  \]
  This contradicts the assumed extremeness of \(A\). Therefore, we can
  conclude that the extreme points of the set of all reasonable,
  efficient allocations are   contained within the set of special allocations. 
\end{proof}

Note, it is clear from the definitions that each special allocation is
an extreme point. So, as a result, we may state the following.

\begin{theorem}
  The extreme points of the set of all reasonable, efficient
  allocations are precisely the special allocations.
\end{theorem}

\subsection{Reasonable efficient allocations and the convex hull of
  the special allocations}

With all of the machinery we have established, we can now prove a nice
property of the special allocations, namely any reasonable,
efficient allocation can be written as a
positive (generalized) linear combination of the special allocations. 
We now prove the main result.

\begin{theorem}\label{thm:saavConvexHullRAAV}
  Any reasonable, efficient allocation can be written as a convex
  combination of the special allocations, more strongly, 
  \begin{center}
    An allocation is reasonable and efficient if and only if the
    allocation lies within the convex hull of the special allocations.
  \end{center}
\end{theorem}

\begin{proof}
  We prove this for the matrix \(A\) of an allocation
  \(\phi\).  If \(A\) lies within the convex hull of all special
  allocation matrices, it is
  clear that our proposition is true based on the inherent properties
  of special allocations explored in \cref{prop:convexComboMaps}, in
  particular.

  Conversely, suppose we have the set of all reasonable, efficient allocations,
  \(\mathcal{R}\), say. Our goal is to show that
  \(\mathcal{R}\) is identical to the convex hull of the special allocations.
  First, we know that the set of all reasonable, efficient allocations is compact, as
  \(\mathcal{R}\) is finite dimensional over \(\mathbb{R}\), and each
  entry of the matrix of a reasonable, efficient allocation is bounded (by \(-1\) and
  \(1\) inclusive via \cref{thm:signRowElt}.) 
  As a result, we have a closed and bounded set, and
  under our conditions, this results in a compact set via the
  Heine-Borel theorem.
  From our prior exploration in
  \cref{prop:convexComboMaps}, we know that \(\mathcal{R}\)
  is convex. Via \cref{lem:raavExtremePts}, we know that the
  extreme points of \(\mathcal{R}\) are precisely the set of all
  special allocations, \(\mathcal{S}\), say. Putting these facts together, we
  may conclude, by  \cref{thm:krein-Millman}
  (\nameref{thm:krein-Millman}) that 
  \begin{align*}
    \mathcal{R}&=\overline{\text{co}}\left(\text{ex}\mathcal{R}\right)\\
               &=\overline{\text{co}}\left(\mathcal{S}\right)
  \end{align*}
  This is nearly what we wish to show. Note, 
  \(\mathcal{S}\) is a finite set and is also bounded, 
  hence closed, thus one can make the last conclusion, that the
  closure of 
  the convex hull is the convex hull itself, via the theorems and
  definitions in \cref{sec:metr-topol-vect}.
  Therefore,
  \begin{align*}
    \mathcal{R}&=\overline{\text{co}}\left(\mathcal{S}\right)\\
               &={\text{co}}\left(\mathcal{S}\right),
  \end{align*}
  and we now have what we set out to prove. Any reasonable,
  efficient
  allocation lies within the convex hull of the special allocations.
\end{proof}

We now have both directions of our proof, and so, we can think of the
special allocations as a spanning set of sorts for the set of
reasonable, efficient allocations. More specifically, any
such allocation is in the cone made up of the special
allocations. Viewing things this way lets us prove several more
results.

\subsection{Consequences of the main result}

With the knowledge gained through
\cref{thm:saavConvexHullRAAV}, we can now prove even more
properties of the reasonable, efficient allocations. 

\begin{theorem}
  Given a reasonable, efficient allocation matrix \(A\), the sum of each interior column
  (not first or last) in absolute value is \(2\).
\end{theorem}

\begin{proof}
  This is trivial to see in the case of the special allocation
  matrices. To obtain the general result, we must see that the result
  holds for a convex combination of two reasonable, efficient allocations
  that satisfy the condition already.\footnote{As any reasonable
    allocation can be written as a (generalized) convex combination,
    this gives us our result.}
  Observe, the column-wise sum, of two such allocation
  matrices, \(P\) and
  \(Q\) for column \(S\).
  \({\left(
    P_{1,S} , P_{2,S} , \ldots , P_{n-1,S} , P_{n,S}
  \right)}^t\)
  and, similarly
  \({\left(
    Q_{1,S} , Q_{2,S} , \ldots , Q_{n-1,S} , Q_{n,S}
  \right)}^t\).
  By our assumption,
  \[
    \sum_{i=1}^n \lvert P_{i,S} \rvert = 2
  \]

  \[
    \sum_{i=1}^n \lvert Q_{i,S} \rvert = 2
  \]
  for \(S \neq \varnothing\) and \(S\neq N\). 
  Now, taking the convex combination  of the columns, we obtain
  \[
    {\left(
      t P_{1,S}+(1-t) Q_{1,S}  , \ldots , t P_{n,S}+(1-t) Q_{n,S} 
    \right)}^t.
  \]
  We note
  \begin{align*}
    \sum_{i=1}^n \lvert t P_{i,S}+(1-t) Q_{i,S} \rvert & \leq
    \sum_{i=1}^n t\lvert P_{i,S} \rvert + (1-t) \lvert Q_{i,S}
                                                         \rvert\\ 
    & = t \sum_{i=1}^n \lvert P_{i,S} \rvert + (1-t)\sum_{i=1}^n \lvert
      Q_{i,S} \rvert\\
                                                       &= t \cdot 2 +
                                                         (1-t) \cdot 2\\
    &= 2
  \end{align*}
  by the triangle inequality and properties of the absolute
  value.  Now, one must argue that this inequality is necessarily an
  equality.  Recall, the triangle inequality is an equality if both
  numbers are non-positive or non-negative. By
  \cref{lem:orderSigns} we see that \(P_{i,S}\) and
  \(Q_{i,S}\) necessarily are both non-positive or non-negative. So
  the triangle inequality is a triangle equality, and we have the
  equality we desire. 
\end{proof}

\begin{theorem}
  The sum of each row of an allocation's matrix in absolute value is
  \(2\).
\end{theorem}

\begin{note}
  We could have proved this without our main result, however, it is
  presented here with similar in spirit results.
\end{note}

\begin{proof}
  Recall, via \cref{lem:orderSigns}, the elements summed in
  \cref{thm:partialRowSum1} are all non-negative. Similarly,
  the sum in the corollary is of all non-positive numbers. As a
  result,
  \[
    \sum_{S: i\notin S} \lvert A_{i,S \cup \{i\}} \rvert= 1
  \]
  and
  \[
    \sum_{S: i\notin S} \lvert A_{i,S} \rvert= 1.
  \]
  Hence,
  \[
    \sum_{S} \lvert A_{i,S}\rvert= 2.
  \]
\end{proof}

\begin{remark}
  We can now recognize ``un-reasonable'' allocations quite easily. If
  the row sums of the absolute value of the elements of the matrix
  of the allocation are not equal to 2 and if the column sums of an
  interior column
  are not equal to 2, we can immediately notice it is unreasonable.
\end{remark}

\section{Consequential Parallel Results}
\label{cha:cons-parall-results}

With slightly different assumptions, we can get the same results, in a
broader context.

\subsection{Superadditivity as a replacement for monotonicity}

To get similair results for superadditive functions, we need only add
the following axioms,

\begin{axiom}\label{ax:valueNothing}
  The value of \(v(\varnothing) = 0\). 
\end{axiom}

\begin{remark}
  This is of course true for superadditve functions, and is not so
  much an axiom, but a consequence of the definition of superadditivity.
  If \(v(\varnothing)>0\), we reach a contradiction, as
  \(v(S)=v(S\cup \varnothing) \geq v(S)+v(\varnothing)\).
\end{remark}

Typically, one would assume \cref{ax:valueNothing}, if one wants to
divide all ``produced'' among the players of the game. This was
mentioned when we first defined efficiency in
\cref{def:newEfficiency}.

\begin{axiom}\label{ax:rowSum0}
  The sum of all elements in each row is of the matrix of the
  allocation \(\phi\) is 0.
\end{axiom}

\begin{note}
  The choice of the first column is arbitrary, due to the fact that
  \(v(\varnothing)=0\) via \cref{ax:valueNothing}, or its following
  remark. Thus, without loss of generality, \cref{ax:rowSum0} always
  holds, as we can pick the value in the first column to make the
  sum work.
\end{note}

To first proceed, let us get an idea what we can do with superadditive
characteristic functions.

\begin{proposition}\label{prop:superBinarySpan}
  The set of all superadditive binary characteristic functions
  (superadditive simple games) form a spanning set for
  the monotone binary characteristic function with
  \(v(\varnothing)=0\). 
\end{proposition}

\begin{proof}
  If we look at the linear span of superadditive
  binary \(v\), they do form a spanning set for all monotone \(v\)
  with \(v(\varnothing) = 0\). We can see this by viewing them in the
  following way. Consider all of the ``monotone'' set chains,
  \[
    \varnothing \subset \{i\}\subset \{i,j\}\subset \{i,j,k\} \subset \ldots
    \subset N,
  \]
  adding a single player in each subsequent set in the chain. This set chain, can
  of course be associated with a special allocation, but it can also be
  associated with a superadditive characteristic function
  \(v_{sa}\),\footnote{This is not one of our previously named
    superadditive characteristic functions.}
  following the steps below. 
  Place a 1 in all of the places associated with each non-empty set in
  the chain, and a 0 in all others. This is trivially superadditive,
  as
  \[
    v_{sa}(S) + v_{sa}(T) \leq v_{sa}(S+T)
  \]
  by construction for \(S\cap T = \varnothing\). Notice also, we can
  truncate this \(v_{sa}\), starting the assignment of ones at any
  point midway through the set chain still yields a superadditive
  characteristic function.  Recall, superadditive implies
  monotone in our considerations.
  However, the reverse direction is easily seen to be false, there are
  certainly monotone characteristic functions that are not superadditive. To see that
  these are within the linear span of the superadditive characteristic
  functions, we
  can follow a constructive process detailed below.
  Given a monotone \(v\), there is at least one set with smallest
  cardinality. Take all of the chains starting with these minimal sets,
  and add together their associated characteristic functions,
  call it \(v_{k}\), say. Notice, if there were more than one set of
  smallest cardinality, \(v_k\) is no longer a binary vector, or
  simple game. To fix
  this, we subtract off a truncation of the characteristic function associated with
  our set chains, starting where our vector has a place containing a
  value more than 1, being careful that we do not subtract anything
  from a place with a 1 in it, as this would break our monotonicity.
  We continue this process, for the vector \(v-v_k\), updating \(v_k\)
  as we proceed, and after finitely many steps \(v=v_k\) and we are done.
\end{proof}

\begin{definition}[\cref{def:reasonable}, redux]\label{def:reasonableRedux}
  Assuming \cref{ax:rowSum0}, An allocation is \emph{reasonable for
    superadditive characteristic functions} if
  \begin{equation}\label{eq:reiteratedBounds}
    \min_{S, i\notin S}\left\{ v\left(S \cup\{i\}\right) -v(S)\right\}
    \leq \phi_i(N;v) \leq \max_{S, i \notin
      S}\left\{v\left(S \cup \{i\}\right) - v(S)\right\}
  \end{equation}
  is satisfied for all superadditive \(v\).
\end{definition}

\begin{note}
  Recall for an allocation \(\phi\) with matrix \(A\), we have \(\phi_i(N;v)=A_i\cdot v\).
\end{note}

\begin{proposition}\label{prop:eff-superadd-monotone}
  If a map is efficient for all superadditive characteristic functions \(v\), and we
  assume \cref{ax:rowSum0} holds, then it is efficient for all
  monotone characteristic functions.
\end{proposition}

\begin{proof}
  We prove this by viewing the allocation as a matrix. To deal with
  efficiency, we 
  look at the sums of the column elements, and ensure they add up to
  \[
    {\left(
      -1,0,\ldots,0,1
    \right)}
  \]
  following our note following \cref{lem:rowwisesum}.
  This, as seen in \cref{lem:rowwisesum}, depends on this being
  true for all \(v\), which we can shorten to monotone \(v\)  thanks to our
  prior work. This time, however, we can consider only superadditive
  \(v\) via \cref{prop:superBinarySpan}.
  So, following the results in the proof
  of \cref{lem:rowwisesum}, we obtain all the sums of column
  elements except the first \(-1\). This \(-1\) is given to
  us by the assumption that the row-wise sum is 0. Hence, as the sum
  of all the rows, save the first element in each row is \(1\),
  this forces the entries in the first column to sum to \(-1\). Thus, we
  have efficiency via the superadditive \(v\) only.
\end{proof}

If we have a matrix of a reasonable allocation, we can also find the 
pairings while checking only the superadditive vectors by
the following \lcnamecref{prop:superaddPairing}.

\begin{proposition}\label{prop:superaddPairing}
  Given an allocation \(\phi\), reasonable for superadditive
  characteristic functions, with matrix \(A\), and
  assuming \cref{ax:rowSum0}, the row-wise pairing in the matrix can be determined
  by using only superadditive characteristic
  functions.
\end{proposition}

\begin{proof}
  Certainly, as we have seen previously, the row-wise pairing of
  elements can be determined by superadditive characteristic functions and their
  truncations, save the 
  \begin{equation}\label{eq:lastPairFound}
    A_{i,\varnothing} = -A_{i,\{i\}}
  \end{equation}
  pair. Thus, we need only check a subset of the
  superadditive characteristic functions to obtain all but the \(n\)
  pairings mentioned in \cref{eq:lastPairFound}.
  Following our prior method of proof, we need the vector
  \({\left[
    1,1,\ldots,1,1
  \right]}^t\) to obtain the last pairings above. This
  is not superadditive, as \(v(\varnothing)>0\).
  However, this vector is simply a convenient way to ensure that
  the sum of each row is 0. Supposing \cref{ax:rowSum0} holds,
  observe that the sum of all of the elements in the row is 0. However,
  all of the other elements in each row sum to 0 in pairs, except \(A_{i,\varnothing}\) and \(A_{i,\{i\}}\). Thus, we
  immediately gain the final pairing, for when we take the row sum,
  it collapses to the two elements,
  \[
    A_{i,\varnothing} + A_{i,\{i\}} = 0,
  \]
  we need only re-arrange and obtain the final pairings, 
  \[
    A_{i,\varnothing} = -A_{i,\{i\}}.
  \]
\end{proof}

\begin{note}
  The pairing alone is not sufficient to show reasonableness, we would
  additionally need that
  \[
    \sum_S A_{i,S\cup\{i\}} = 1
  \]
  for \(S\) without \(i\) and \(A_{i,S\cup\{i\}} \geq 0\) for the same
  \(S\). Then, certainly \(A\) is reasonable.
\end{note}

Notice, with no modifications whatsoever that
\cref{lem:raavExtremePts} holds. Additionally, with the background
above, we have \cref{thm:saavConvexHullRAAV} as well, replacing
reasonable with reasonable for superadditive characteristic functions, as the argument
does not depend on superadditive or monotone \(v\) in the slightest.

\begin{proposition}\label{prop:superreas-convex-combo}
  If \(\phi\) is reasonable and efficient for superadditive
  characteristic functions, then it is a convex combination of the
  special allocations.
\end{proposition}

\begin{proof}
  This is mainly a direct consequence of
  \cref{lem:raavExtremePts,prop:superaddPairing,prop:eff-superadd-monotone}.
  Suppose we have a matrix \(A\) of the allocation \(\phi\) that is reasonable for
  superadditive characteristic functions. By
  \cref{prop:eff-superadd-monotone} we know that the same efficiency
  constrains are satisfied. Further, by \cref{prop:superaddPairing} we
  have the pairings we seek. Finally, via \cref{lem:raavExtremePts} we
  see that the extreme points of the reasonable, efficient
  allocations are the special allocations. To complete the result, we apply
  \cref{thm:saavConvexHullRAAV}, with the prior results on reasonable
  for superadditive characteristic functions \(v\) and the proof is complete.
\end{proof}

With this result, we note the following \lcnamecref{cor:superreasreas}.

\begin{corollary}[to \cref{prop:superreas-convex-combo}]\label{cor:superreasreas}
  Given an efficient allocation, reasonable for superadditive
  characteristic functions implies reasonableness.
\end{corollary}

\begin{indeed}
  We can trivially observe that if something lies within the convex
  combination of the special allocations, then it is reasonable by the vanilla
  version of \cref{thm:saavConvexHullRAAV}.
\end{indeed}

\begin{note}
  This is quite nontrivial. If one attempts to prove this fact from first
  principles it is difficult, if not impossible.
\end{note}

\begin{proposition}
  Assuming \cref{ax:rowSum0}, reasonableness implies reasonable for
  superadditive characteristic functions.
\end{proposition}

\begin{proof}
  This is trivially the case. If one satisfies
  reasonableness for all monotone \(v\),
  \cref{eq:reiteratedBounds} is certainly satisfied for all
  superadditive \(v\).
\end{proof}

We conclude by distilling our results into the
following \namecref{thm:reas-superreas}.

\begin{theorem}\label{thm:reas-superreas}
  An efficient allocation is reasonable if and only if it is
  reasonable for superadditive characteristic functions.
\end{theorem}

\subsection{Further exploration}

We note that our results here can be generalized further, in both
directions. Namely, all of the results we have seen can be made less
stringent. In all of our reasonability discussions, we have used only a
small set of superadditive characteristic functions, 

\begin{align*}
  v^S_a(T) &=
             \begin{cases}
               1 & \text{if } S \subset T\\
               0 & \text{else}
             \end{cases}\\
  v^S_b(T) &=
             \begin{cases}
               1 & \text{if } S \subsetneq T\setminus\{i\}\\
               0 & \text{else}
             \end{cases}\\
  v^S_c(T) &=
             \begin{cases}
               1 & \text{if } S \subsetneq T\\
               0 & \text{else}
             \end{cases}.
\end{align*}
Our results hold if we are reasonable and efficient for the set \(\mathcal{V}_{abc}\) 
containing all vectors of this type.

Further, as long as our general set of vectors contains this set of
vectors, we can establish a version of reasonableness, and obtain the
results once again.

\subsection{Concluding remarks}

As we discussed previously, one of the important ideas of cooperative
game theory, and more generally, a way to fairly determine power or
distribute gains, is the Shapley value. As such, an understanding of
the setting surrounding the value, and the axiomatic assumptions is
necessary. Within this paper, we discovered in the general,
non-probabilistic context, that alternative assumptions still give us
a robust structure.

Even without the Shapley value's uniqueness, the structure within
gives us some insight on 
how a reasonable, efficient allocation is constructed. This insight
leads us to other possible values, parallel to Shapley's, offering a
structure for an alternative way of distributing the gains of
collaboration.

\bibliographystyle{plainnat}
\bibliography{../universal-st}

\end{document}